\documentclass[11pt]{article}
\usepackage{epsf}
\usepackage{amsmath}
\usepackage{epsfig}
\usepackage{times}
\usepackage{amssymb}
\usepackage{amsthm}
\usepackage{setspace}
\usepackage{cite}

\usepackage{algorithmic}  % This package provides an algorithmic environment fo describing algorithms.
\usepackage{algorithm}

\usepackage{shadow}
\usepackage{fancybox}
\usepackage{fancyhdr}

\def\w{{\bf w}}

\def\y{{\bf y}}

\def\x{{\bf x}}

\def\x{{\mathbf x}}

\def\w{{\bf w}}

\def\x{{\bf x}}
\def\y{{\bf y}}

\def\h{{\bf h}}

\def\be{\begin{equation}}
\def\ee{\end{equation}}
\def\ba{\left[\begin{array}}
\def\ea{\end{array}\right]}

\def\w{{\bf w}}

\def\x{{\bf x}}
\def\y{{\bf y}}

\def\1{{\bf 1}}

\def\g{{\bf g}}
\def\0{{\bf 0}}

\def\erfinv{\mbox{erfinv}}
\def\htheta{\hat{\theta}}

\newtheorem{theorem}{Theorem}

\newtheorem{lemma}{Lemma}

\begin{document}

\begin{singlespace}

\title {Meshes that trap random subspaces %A tight variant of Gordon's escape through a mesh theorem
%\footnote{ This work was supported in
%part.}
}
\author{
\textsc{Mihailo Stojnic}
\\
\\
{School of Industrial Engineering}\\
{Purdue University, West Lafayette, IN 47907} \\
{e-mail: {\tt mstojnic@purdue.edu}} }
\date{}
\maketitle

\centerline{{\bf Abstract}} \vspace*{0.1in}

In our recent work \cite{StojnicCSetam09,StojnicUpper10} we considered solving under-determined systems of linear equations with sparse solutions.
In a large dimensional and
statistical context we proved results related to performance of a polynomial $\ell_1$-optimization technique
when used for solving such systems. As one of the tools we used a probabilistic result of Gordon \cite{Gordon88}. In this paper we revisit this classic result in its core form and show how it can be reused to in a sense prove its own optimality.

\vspace*{0.25in} \noindent {\bf Index Terms: Random subspaces; linear systems of equations;
$\ell_1$-optimization}.

\end{singlespace}

%%%%%%%%%%%%%%%%%%%%%%%%%%%%%%%%%%%%%%%%%%%%%%%%%%%%%%%%%%%%%%%%%
\section{Introduction}
\label{sec:back}
%%%%%%%%%%%%%%%%%%%%%%%%%%%%%%%%%%%%%%%%%%%%%%%%%%%%%%%%%%%%%%%%%

We start by looking back at the problem that we considered in a series of recent work \cite{StojnicCSetam09,StojnicICASSP09,StojnicUpper10}. It essentially boils down to finding sparse solutions of under-determined systems of linear equations. In a more precise mathematical language we would like to find a $k$-sparse $\x$ such
that
\begin{equation}
A\x=\y \label{eq:system}
\end{equation}
where $A$ is an $m\times n$ ($m<n$) matrix and $\y$ is
an $m\times 1$ vector (here and in the rest of the paper, under $k$-sparse vector we assume a vector that has at most $k$ nonzero
components). Of course, the assumption will be that such an $\x$ exists.

To make writing in the rest of the paper easier, we will assume the
so-called \emph{linear} regime, i.e. we will assume that $k=\beta n$
and that the number of equations is $m=\alpha n$ where
$\alpha$ and $\beta$ are constants independent of $n$ (more
on the non-linear regime, i.e. on the regime when $m$ is larger than
linearly proportional to $k$ can be found in e.g.
\cite{CoMu05,GiStTrVe06,GiStTrVe07}).

A particularly successful technique for solving (\ref{eq:system}) is a linear programming relaxation called $\ell_1$-optimization. (Variations of the standard $\ell_1$-optimization from e.g.
\cite{CWBreweighted,SChretien08,SaZh08}) as well as those from \cite{SCY08,FL08,GN03,GN04,GN07,DG08} related to $\ell_q$-optimization, $0<q<1$
are possible as well.) Basic $\ell_1$-optimization algorithm finds $\x$ in
(\ref{eq:system}) by solving the following $\ell_1$-norm minimization problem
\begin{eqnarray}
\mbox{min} & & \|\x\|_{1}\nonumber \\
\mbox{subject to} & & A\x=\y. \label{eq:l1}
\end{eqnarray}
Due to its popularity the literature on the use of the above algorithm is rapidly growing. We below restrict our attention to two, in our mind, the most influential works that relate to (\ref{eq:l1}).

The first one is \cite{CRT} where the authors were able to show that if
$\alpha$ and $n$ are given, $A$ is given and satisfies the restricted isometry property (RIP) (more on this property the interested reader can find in e.g. \cite{Crip,CRT,Bar,Ver,ALPTJ09}), then
any unknown vector $\x$ with no more than $k=\beta n$ (where $\beta$
is a constant dependent on $\alpha$ and explicitly
calculated in \cite{CRT}) non-zero elements can be recovered by
solving (\ref{eq:l1}). As expected, this assumes that $\y$ was in
fact generated by that $\x$ and given to us.

However, the RIP is only a \emph{sufficient}
condition for $\ell_1$-optimization to produce the $k$-sparse solution of
(\ref{eq:system}). Instead of characterizing $A$ through the RIP
condition, in \cite{DonohoUnsigned,DonohoPol} Donoho looked at its geometric properties/potential. Namely,
in \cite{DonohoUnsigned,DonohoPol} Donoho considered the polytope obtained by
projecting the regular $n$-dimensional cross-polytope $C_p^n$ by $A$. He then established that
the solution of (\ref{eq:l1}) will be the $k$-sparse solution of
(\ref{eq:system}) if and only if
$AC_p^n$ is centrally $k$-neighborly
(for the definitions of neighborliness, details of Donoho's approach, and related results the interested reader can consult now already classic references \cite{DonohoUnsigned,DonohoPol,DonohoSigned,DT}). In a nutshell, using the results
of \cite{PMM,AS,BorockyHenk,Ruben,VS}, it is shown in
\cite{DonohoPol}, that if $A$ is a random $m\times n$
ortho-projector matrix then with overwhelming probability $AC_p^n$ is centrally $k$-neighborly (as usual, under overwhelming probability we in this paper assume
a probability that is no more than a number exponentially decaying in $n$ away from $1$). Miraculously, \cite{DonohoPol,DonohoUnsigned} provided a precise characterization of $m$ and $k$ (in a large dimensional context) for which this happens.

In a series of our own work (see, e.g. \cite{StojnicICASSP09,StojnicCSetam09}) we then created an alternative probabilistic approach which was capable of providing the precise characterization between $m$ and $k$ that guarantees success/failure of (\ref{eq:l1}) when used for finding the $k$-sparse solution of (\ref{eq:system}). The approach was a combination of geometric and purely probabilistic ideas and used bunch of tools from classical probability theory, (most notably a couple of results of Gordon from \cite{Gordon88} that we will revisit in this paper). The following theorem summarizes the results we obtained in e.g. \cite{StojnicICASSP09,StojnicCSetam09}.

\begin{theorem}(Exact threshold)
Let $A$ be an $m\times n$ matrix in (\ref{eq:system})
with i.i.d. standard normal components. Let
the unknown $\x$ in (\ref{eq:system}) be $k$-sparse. Further, let the location and signs of nonzero elements of $\x$ be arbitrarily chosen but fixed.
Let $k,m,n$ be large
and let $\alpha=\frac{m}{n}$ and $\beta_w=\frac{k}{n}$ be constants
independent of $m$ and $n$. Let $\erfinv$ be the inverse of the standard error function associated with zero-mean unit variance Gaussian random variable.  Further,
let all $\epsilon$'s below be arbitrarily small constants.
\begin{enumerate}
\item Let $\htheta_w$, ($\beta_w\leq \htheta_w\leq 1$) be the solution of
\begin{equation}
(1-\epsilon_{1}^{(c)})(1-\beta_w)\frac{\sqrt{\frac{2}{\pi}}e^{-(\erfinv(\frac{1-\theta_w}{1-\beta_w}))^2}}{\theta_w}-\sqrt{2}\erfinv ((1+\epsilon_{1}^{(c)})\frac{1-\theta_w}{1-\beta_w})=0.\label{eq:thmweaktheta}
\end{equation}
If $\alpha$ and $\beta_w$ further satisfy
\begin{equation}
\alpha>\frac{1-\beta_w}{\sqrt{2\pi}}\left (\sqrt{2\pi}+2\frac{\sqrt{2(\erfinv(\frac{1-\htheta_w}{1-\beta_w}))^2}}{e^{(\erfinv(\frac{1-\htheta_w}{1-\beta_w}))^2}}-\sqrt{2\pi}
\frac{1-\htheta_w}{1-\beta_w}\right )+\beta_w
-\frac{\left ((1-\beta_w)\sqrt{\frac{2}{\pi}}e^{-(\erfinv(\frac{1-\hat{\theta}_w}{1-\beta_w}))^2}\right )^2}{\hat{\theta}_w}\label{eq:thmweakalpha}
\end{equation}
then with overwhelming probability the solution of (\ref{eq:l1}) is the $k$-sparse $\x$ from (\ref{eq:system}).
\item Let $\htheta_w$, ($\beta_w\leq \htheta_w\leq 1$) be the solution of
\begin{equation}
(1+\epsilon_{2}^{(c)})(1-\beta_w)\frac{\sqrt{\frac{2}{\pi}}e^{-(\erfinv(\frac{1-\theta_w}{1-\beta_w}))^2}}{\theta_w}-\sqrt{2}\erfinv ((1-\epsilon_{2}^{(c)})\frac{1-\theta_w}{1-\beta_w})=0.\label{eq:thmweaktheta1}
\end{equation}
If on the other hand $\alpha$ and $\beta_w$ satisfy
\begin{multline}
\hspace{-.5in}\alpha<\frac{1}{(1+\epsilon_{1}^{(m)})^2}\left ((1-\epsilon_{1}^{(g)})(\htheta_w+\frac{2(1-\beta_w)}{\sqrt{2\pi}} \frac{\sqrt{2(\erfinv(\frac{1-\htheta_w}{1-\beta_w}))^2}}{e^{(\erfinv(\frac{1-\htheta_w}{1-\beta_w}))^2}})
-\frac{\left ((1-\beta_w)\sqrt{\frac{2}{\pi}}e^{-(\erfinv(\frac{1-\hat{\theta}_w}{1-\beta_w}))^2}\right )^2}{\hat{\theta}_w(1+\epsilon_{3}^{(g)})^{-2}}\right )\label{eq:thmweakalpha}
\end{multline}
then with overwhelming probability there will be a $k$-sparse $\x$ (from a set of $\x$'s with fixed locations and signs of nonzero components) that satisfies (\ref{eq:system}) and is \textbf{not} the solution of (\ref{eq:l1}).
\end{enumerate}
\label{thm:thmweakthr}
\end{theorem}
\begin{proof}
The first part was established in \cite{StojnicCSetam09} and the second one was established in \cite{StojnicUpper10}. An alternative way of establishing the same set of results was also presented in \cite{StojnicEquiv10}.
\end{proof}

We below provide a more informal interpretation of what was established by the above theorem. Assume the setup of the above theorem. Let $\alpha_w$ and $\beta_w$ satisfy the following:

\noindent \underline{\underline{\textbf{Fundamental characterization of the $\ell_1$ performance:}}}

\begin{center}
\shadowbox{$
%\begin{equation}
(1-\beta_w)\frac{\sqrt{\frac{2}{\pi}}e^{-(\erfinv(\frac{1-\alpha_w}{1-\beta_w}))^2}}{\alpha_w}-\sqrt{2}\erfinv (\frac{1-\alpha_w}{1-\beta_w})=0.
%\end{equation}
$}
-\vspace{-.5in}\begin{equation}
\label{eq:thmweaktheta2}
\end{equation}
\end{center}

Then:
\begin{enumerate}
\item If $\alpha>\alpha_w$ then with overwhelming probability the solution of (\ref{eq:l1}) is the $k$-sparse $\x$ from (\ref{eq:system}).
\item If $\alpha<\alpha_w$ then with overwhelming probability there will be a $k$-sparse $\x$ (from a set of $\x$'s with fixed locations and signs of nonzero components) that satisfies (\ref{eq:system}) and is \textbf{not} the solution of (\ref{eq:l1}).
    \end{enumerate}

As mentioned above, to establish the result given in (\ref{eq:thmweaktheta2}) we used a couple of classic probabilistic results from \cite{Gordon88}. In the following section we will recall these results and see how they can be reconnected and in a way optimized.

We organize the rest of the paper in the following way. In Section
\ref{sec:theorems} we introduce and briefly discuss the two theorems from \cite{Gordon88} that we plan to revisit in this paper, while in
Section
\ref{sec:exact} we create the mechanism for optimizing the second of the theorems in certain scenarios. Finally, in Sections \ref{sec:discussion} and \ref{sec:conc} we discuss obtained results.

%%%%%%%%%%%%%%%%%%%%%%%%%%%%%%%%%%%%%%%%%%%%%%%%%%%%%%%%%%%%%%%%%
\section{Key theorems}
\label{sec:theorems}
%%%%%%%%%%%%%%%%%%%%%%%%%%%%%%%%%%%%%%%%%%%%%%%%%%%%%%%%%%%%%%%%%

In this section we introduce the above mentioned theorems that will be of key importance in our subsequent
considerations.

First we recall the following results from \cite{Gordon88} that relates to statistical properties of certain Gaussian processes.
\begin{theorem}(\cite{Gordon88})
\label{thm:Gordonmesh1} Let $X_{ij}$ and $Y_{ij}$, $1\leq i\leq n,1\leq j\leq m$, be two centered Gaussian processes which satisfy the following inequalities for all choices of indices
\begin{enumerate}
\item $E(X_{ij}^2)=E(Y_{ij}^2)$
\item $E(X_{ij}X_{ik})=E(Y_{ij}Y_{ik})$
\item $E(X_{ij}X_{lk})=E(Y_{ij}Y_{lk}), i\neq l$.
\end{enumerate}
Then
\begin{equation*}
P(\bigcap_{i}\bigcup_{j}(X_{ij}\geq \lambda_{ij}))\leq P(\bigcap_{i}\bigcup_{j}(Y_{ij}\geq \lambda_{ij})).
\end{equation*}
\end{theorem}

Based on the above theorem Gordon then went further and proved a more specific type of result now widely known as ``Escape through a mesh" theorem. The result essentially looks at a particular class of Gaussian processes and connects them with the geometry of random subspaces and their intersections with given fixed subsets of high-dimensional unit spheres.
\begin{theorem}(\cite{Gordon88} Escape through a mesh)
\label{thm:Gordonmesh} Let $S$ be a subset of the unit Euclidean
sphere $S^{n-1}$ in $R^{n}$. Let $Y$ be a random
$(n-m)$-dimensional subspace of $R^{n}$, distributed uniformly in
the Grassmanian with respect to the Haar measure. Let
\begin{equation}
w_D(S)=E\sup_{\w\in S} (\g^T\w) \label{eq:widthdef}
\end{equation}
where $\h$ is a random column vector in $R^{n}$ with i.i.d. ${\cal
N}(0,1)$ components. Assume that
$w_D(S)<\left ( \sqrt{m}-\frac{1}{4\sqrt{m}}\right )$. Then
\begin{equation}
P(Y\cap S=0)>1-3.5e^{-\frac{\left (
\sqrt{m}-\frac{1}{4\sqrt{m}}-w_D(S) \right ) ^2}{18}}.
\label{eq:thmesh}
\end{equation}
\end{theorem}
\textbf{Remark}: Gordon's original constant $3.5$ was substituted by
$2.5$ in \cite{RVmesh}. Both constants are not subject of our detailed considerations. However, we do mention in passing that to the best of our knowledge it is an open problem to determine the exact value of this constant as well as to improve and ultimately determine the exact value as well of somewhat high constant $18$.

In a more informal language, what Theorem \ref{thm:Gordonmesh} manages to create is a route to connect the location of a low-dimensional random subspace with respect to a given body and a seemingly simple quantity $w(S)$. Then as long as one can get a handle of $w(S)$ and dimensions are large enough one can get a pretty good feeling if a random flat will hit or miss the given body $S$. There are a couple of restrictions, though. What we call a body in an informal way is not really a body but rather a subset of the unit $n$-dimensional sphere and what we call a random flat is not really ``just" a random flat but actually a subspace chosen uniformly randomly from the Grassmanian (one can think of it as a uniformly random choice among all subspaces of dimension $n-m$). We believe that it is easier to get a real feeling of the power of Gordon's results if one for a moment leaves technicalities out of the picture and instead views things in a more informal way.

Along the same lines, in our opinion, to fully understand the miraculous importance of Theorem \ref{thm:Gordonmesh} it is maybe a good starting point to have a firm hold of understanding of the original geometric question that it answers. The question is incredibly simple: there is a set $S$ which is a subset of sphere $S^{n-1}$ in $R^n$. One then generates a uniformly random subspace (as we said above, in this paper, when we talk about uniformly random spaces/subspaces we of course view such a randomness in a Grassmanian sense) of dimension say $n-m$ (where of course $m\geq 0$) and wonders how likely is that such a subspace will intersect with $S$. One simple example that could help visualizing these high-dimensional geometric concepts would be to take $n=3$ and look at a spherical cap of the sphere $S^2$ in $R^3$. Then one can chose say $n-m=1$ and basically wonder how likely is that a random line through the origin would intersect such a spherical cap. Of course when $S$ is a spherical cap the answer is simple and can be obtained through a simple geometric consideration as the ratio of the spherical cap's area and the area of the entire unit sphere. On the other hand, geometrically speaking, it is immediately clear how much harder the question becomes if $S$ is not a spherical cap and $n$ and $n-m$ are large.

If one then looks back at the original question, which, as we discussed above, is purely geometrical, it seems almost unbelievable (at least a priori) that it can be transformed to a purely analytical problem. The incredible contribution of Theorem \ref{thm:Gordonmesh} is exactly in its success to create such a transformation and to effectively connect this geometric question on one side and the properties of Gaussian processes on the other side. The idea of moving everything to the analysis terrain is great on its own, however what is more astonishing is that often one can actually accomplish it. Still, when one moves to the analysis terrain there are several questions one should be able to tackle (the problem may just seem a but easier when transferred to the analysis terrain, but nobody guarantees that it is actually easy!). The two questions that we found most pressing are:
1) Can one get a handle of $w_D(S)$ for any $S$?
2) Roughly speaking, the theorem only specifies what will happen if $w_D(S)<\sqrt{m}$. Is there a definite answer as to what will happen if $w_D(S)\geq \sqrt{m}$?

When it comes to answering the first question it doesn't seem that the answer would be yes. Still, experience says that for many ``practical" sets $S$ one can actually handle $w_D(S)$ (see, e.g. \cite{StojnicCSetam09,StojnicCSetamBlock09}). Even if computing the exact value of $w_D(S)$ may not be feasible there are possible alternatives. For example, one can try to bound $w_D(S)$ and in a way provide at least some kind of answer to the original geometric question. On the other hand, when it comes to the second question one could envision two possible scenarios. Assuming that the answer to question 1) is no, one then may start looking at particular sets $S$ and then wonder which are the sets $S$ so that $w_D(S)$ can be handled. Then the first scenario would be to look at those $S$ for which $w_D(S)$ would not be computable. Then even if one can give a definite answer to question 2) the whole concept would appear as a raw theory without final analytical concreteness. The second scenario would on the other hand relate to those $S$ for which $w_D(S)$ can be computed. This scenario is actually probably the first next direction for possible further studies of Theorem \ref{thm:Gordonmesh}. In the following section we look at this very same scenario and observe that for certain $S$ one can actually provide a definite answer to question 2.

%%%%%%%%%%%%%%%%%%%%%%%%%%%%%%%%%%%%%%%%%%%%%%%%%%%%%%%%%%%
\section{Revisiting Escape through a mesh theorem} \label{sec:exact}
%%%%%%%%%%%%%%%%%%%%%%%%%%%%%%%%%%%%%%%%%%%%%%%%%%%%%%%%%%%

In the first part of this section we will look at a couple of technical details that relate to quantities from Theorem \ref{thm:Gordonmesh}. We first revisit $w_D(S)$. As stated in Theorem \ref{thm:Gordonmesh}, $w_D(S)$ is given as
\begin{equation}
w_D(S)=E\sup_{\w\in S} (\g^T\w). \label{eq:widthdef0}
\end{equation}
To be a bit more specific we will assume that set $S$ can be described through a functional equation, i.e we will say that
\begin{equation}
S=\{\w| ||\w||_2=1, f(\w)\leq 0\}.\label{eq:setSdef}
\end{equation}
We will then accordingly replace $w(S)$ by
\begin{eqnarray}
w_D(f)=E\sup & & (\g^T\w) \nonumber \\
\mbox{subject to} & & \|\w\|_2= 1\nonumber \\
& & f(\w)\leq 0. \label{eq:widthdef1}
\end{eqnarray}

%%%%%%%%%%%%%%%%%%%%%%%%%%%%%%%%%%%%%%%%%%%%%%%%%%%%%%%%%%%
\subsection{Deterministic view} \label{sec:detview}
%%%%%%%%%%%%%%%%%%%%%%%%%%%%%%%%%%%%%%%%%%%%%%%%%%%%%%%%%%%

Clearly to gain a complete control over $w_D(f)$ (and basically $w_D(S)$) one ultimately has to consider its random origin. However, before going through the randomness of the problem and we will try to provide a more information about $w_D(f)$ (and a couple other deterministic quantities that will be introduced below) on a deterministic level. Along these lines, to distinguish between deterministic and random portions of nature of $w_D(f)$ (i.e. $w_D(f)$) we will introduce quantity $w(f,\g)$ as
\begin{eqnarray}
w(f,\g)=\sup & & (\g^T\w) \nonumber \\
\mbox{subject to} & & \|\w\|_2= 1\nonumber \\
& & f(\w)\leq 0. \label{eq:widthdef1det}
\end{eqnarray}
Then clearly
\begin{equation}
w_D(S)=w_D(f)=E w(f,\g). \label{eq:widthdef2}
\end{equation}
As mentioned above, for the time being we will focus on $w_D(f,\g)$.
Also, to make the presentation easier we will assume that the $\sup$ in (\ref{eq:widthdef}), (\ref{eq:widthdef0}), (\ref{eq:widthdef1}), and (\ref{eq:widthdef1det}) can be replaced with a $\max$ (also for all other occasions in the paper where a $\sup$ may appear as more precise we will assume that scenarios are such that a $\max$ can replace it). Then (\ref{eq:widthdef1det}) can be rewritten as
\begin{eqnarray}
w(f,\g) =  \max & & \g^T\w \nonumber \\
\mbox{subject to} & &  \|\w\|_2= 1\nonumber \\
& & f(\w)\leq 0.\label{eq:widthdef1det1}
\end{eqnarray}
Transforming (\ref{eq:widthdef1det1}) a bit further one gets
\begin{eqnarray}
w(f,\g) =  -\min & & -\g^T\w \nonumber \\
\mbox{subject to} & &  \|\w\|_2= 1\nonumber \\
& & f(\w)\leq 0.\label{eq:widthdef1det2}
\end{eqnarray}
Using a Lagrangian multiplier one can move constraint on $f(\w)$ into the objective
\begin{equation}
w(f,\g) =  -\min_{||\w||_2=1}\max_{\lambda\geq 0} -\g^T\w+\lambda f(\w).\label{eq:widthdef1det3}
\end{equation}
One then easily has
\begin{equation}
w(f,\g) \leq  -\max_{\lambda\geq 0}\min_{||\w||_2=1} -\g^T\w+\lambda f(\w),\label{eq:widthdef1det4}
\end{equation}
and
\begin{equation}
w(f,\g) \leq  \min_{\lambda\geq 0}\max_{||\w||_2=1} \g^T\w-\lambda f(\w).\label{eq:widthdef1det5}
\end{equation}
We will now leave the deterministic portion of the analysis of $w(S)$ (or to be more precise the analysis of $w(f,\g)$) for a moment and switch to consideration of a seemingly different optimization problem. Namely we will consider the following deterministic optimization problem
\begin{eqnarray}
\tau(f,A) =  \min & & f(\w) \nonumber \\
\mbox{subject to} & &  A\w=0 \nonumber \\
& & \|\w\|_2= 1,\label{eq:taudef1det1}
\end{eqnarray}
and through it we will introduce a new quantity $\tau(f,A)$. This quantity will be in a way an ``almost" counterpart to $w(f,\g)$. At this point the purpose of introducing such a quantity may not be clear. However, as we progress further it will become more apparent what its meaning is and why we introduced it. Here we only mention roughly that $\tau(f,A)$ can be thought of as an indicator that subspace of
$\w$'s, $A\w$, and the unit sphere $\|\w\|_2=1$ have an intersection that is also contained in $S$. Namely, if $\tau(f,A)<0$ then indeed there is a $\w$ such that $A\w=0$, $\|\w\|_2=1$, and $f(\w)<0$. However by the definition of $S$ from (\ref{eq:setSdef}) such a $\w$ is actually in $S$. On the other hand if $\tau(f,A)\geq0$ there is no $\w$ such that $A\w=0$, $\|\w\|_2=1$, and $f(\w)<0$ and automatically the intersection of subspace $A\w$ and the unit sphere $\|\w\|_2=1$ is missing set $S$.

Going back to (\ref{eq:taudef1det1}) and using again the Lagrangian multipliers one can then move the subspace constraint into the objective
\begin{equation}
\tau(f,A) =  \min_{||\w||_2=1}\max_{\nu} f(\w)+\nu^TA\w.\label{eq:taudef1det2}
\end{equation}
Now, we will assume that the structure of set $S$ is determined by a function $f(\w)$ for which it also holds
\begin{eqnarray}
\tau(f,A) & = & \min_{||\w||_2=1}\max_{\nu} f(\w)+\nu^TA\w\nonumber \\
& = &  \max_{\nu} \min_{||\w||_2=1}f(\w)+\nu^TA\w.\label{eq:taudef1det3}
\end{eqnarray}
In fact, as it will be clear from the subsequent analysis, the property that we will mostly utilize is actually the sign of $\tau(f,A)$. Having that in mind one can actually relax a bit requirement (\ref{eq:taudef1det3})
\begin{eqnarray}
\mbox{sign}(\tau(f,A)) & = & \mbox{sign}(\min_{||\w||_2=1}\max_{\nu} f(\w)+\nu^TA\w)\nonumber \\
& = &  \mbox{sign}(\max_{\nu} \min_{||\w||_2=1}f(\w)+\nu^TA\w).\label{eq:taudef1det31}
\end{eqnarray}
Clearly, (\ref{eq:taudef1det3}) or (\ref{eq:taudef1det31}) will not hold for any $f(\w)$ and any $A$. However, we will assume that there are $f(\w)$ and $A$ for which they will hold. After rearranging (\ref{eq:taudef1det3}) a bit we have
\begin{equation}
-\tau(f,A) = \min_{\nu} \max_{||\w||_2=1} -f(\w)-\nu^TA\w,\label{eq:taudef1det4}
\end{equation}
and after rearranging (\ref{eq:taudef1det31}) a bit we have
\begin{equation}
-\mbox{sign}(\tau(f,A)) = \mbox{sign}(\min_{\nu} \max_{||\w||_2=1} -f(\w)-\nu^TA\w).\label{eq:taudef1det41}
\end{equation}
At this point one should note that while quantities $w(f,\g)$ and $\tau(f,A)$ are random, so far they have been treated as deterministic. In other words, we viewed them as functions of a fixed pair $(\g,A)$. Moreover, they are in a good enough shape that we can switch to a probabilistic portion of their analysis. Probabilistic portion of the analysis will essentially contain an analysis that will determine typical behavior of these two quantities when components of $\g$ and $A$ are i.i.d. standard normals.

%%%%%%%%%%%%%%%%%%%%%%%%%%%%%%%%%%%%%%%%%%%%%%%%%%%%%%%%%%%
\subsection{Probabilistic view} \label{sec:probview}
%%%%%%%%%%%%%%%%%%%%%%%%%%%%%%%%%%%%%%%%%%%%%%%%%%%%%%%%%%%

To obtain a probabilistic view on quantities $w(f,\g)$ and $\tau(f,A)$ we will invoke the results of Theorem \ref{thm:Gordonmesh1}. We will do so through the following lemma which is slightly modified Lemma 3.1 from \cite{Gordon88} (Lemma 3.1 is a direct consequence of Theorem \ref{thm:Gordonmesh1} and the backbone of the escape through a mesh theorem).
\begin{lemma}
Let $A$ be an $m\times n$ matrix with i.i.d. standard normal components. Let $\g$ and $\h$ be $n\times 1$ and $m\times 1$ vectors, respectively, with i.i.d. standard normal components. Also, let $g$ be a standard normal random variable. Then
\begin{equation}
P(\min_{\nu\in R^n\setminus 0}\max_{\|\w\|_2=1}(-\nu^T A\w +\|\nu\|_2 g-\zeta_{\w,\nu})\geq 0)\geq P(\min_{\nu\in R^n\setminus 0}\max_{\|\w\|_2=1}(\|\nu\|_2\g^T\w+\h^T\nu-\zeta_{\w,\nu})\geq 0).\label{eq:problemma}
\end{equation}\label{eq:taulemma}
\end{lemma}
\begin{proof}
The proof is exactly the same as is the one of Lemma 3.1 in \cite{Gordon88}.
\end{proof}
Let $\zeta_{\w,\nu}=\epsilon_{5}^{(g)}\sqrt{n}\|\nu\|_2+f(\w)$ with $\epsilon_{5}^{(g)}>0$ being an arbitrarily small constant independent of $n$. We will first look at the right-hand side of the inequality in (\ref{eq:problemma}). The following is then the probability of interest
\begin{equation}
P(\min_{\nu\in R^n\setminus 0}\max_{\|\w\|_2=1}(\|\nu\|_2\g^T\w+\h^T\nu-\epsilon_{5}^{(g)}\sqrt{n}\|\nu\|_2-f(\w))\geq 0).\label{eq:probanal0}
\end{equation}
After pulling out $\|\nu\|_2$ one has
\begin{equation*}
\hspace{-.7in}P(\min_{\nu\in R^n\setminus 0}\max_{\|\w\|_2=1}(\|\nu\|_2\g^T\w+\h^T\nu-\epsilon_{5}^{(g)}\sqrt{n}\|\nu\|_2-f(\w))\geq 0)=P(\min_{\nu\in R^n\setminus 0}\max_{\|\w\|_2=1}(\|\nu\|_2(\g^T\w+\frac{\h^T\nu}{\|\nu\|_2}-\epsilon_{5}^{(g)}\sqrt{n}-\frac{f(\w)}{\|\nu\|_2}))\geq 0)
\end{equation*}
and then easily
\begin{equation*}
\hspace{-.5in}P(\min_{\nu\in R^n\setminus 0}\max_{\|\w\|_2=1}(\|\nu\|_2\g^T\w+\h^T\nu-\epsilon_{5}^{(g)}\sqrt{n}\|\nu\|_2-f(\w))\geq 0)=P(\min_{\nu\in R^n\setminus 0}\max_{\|\w\|_2=1}(\g^T\w+\frac{\h^T\nu}{\|\nu\|_2}-\epsilon_{5}^{(g)}\sqrt{n}-\frac{f(\w)}{\|\nu\|_2})\geq 0).
\end{equation*}
Replacing $\|\nu\|_2$ with a scaler $\frac{1}{\lambda_{\nu}}$ and solving the minimization over different $\nu$ with a fixed $\|\nu\|_2$ one obtains
\begin{equation}
\hspace{-.3in}P(\min_{\nu\in R^n\setminus 0}\max_{\|\w\|_2=1}(\|\nu\|_2\g^T\w+\h^T\nu-\epsilon_{5}^{(g)}\sqrt{n}\|\nu\|_2-f(\w))\geq 0)=P(\min_{\lambda_{\nu}\geq 0}\max_{\|\w\|_2=1}(\g^T\w-\lambda_{\nu}f(\w))\geq \|\h\|_2+\epsilon_{5}^{(g)}\sqrt{n}).\label{eq:probanal1}
\end{equation}
Since $\h$ is a vector of $m$ i.i.d. standard normal variables it is rather trivial that $P(\|\h\|_2<(1+\epsilon_{1}^{(m)})\sqrt{m})\geq 1-e^{-\epsilon_{2}^{(m)} m}$ where $\epsilon_{1}^{(m)}>0$ is an arbitrarily small constant and $\epsilon_{2}^{(m)}$ is a constant dependent on $\epsilon_{1}^{(m)}$ but independent of $n$. Then from (\ref{eq:probanal1}) one obtains
\begin{multline}
P(\min_{\lambda_{\nu}\geq 0}\max_{\|\w\|_2=1}(\g^T\w-\lambda_{\nu}f(\w))\geq \|\h\|_2+\epsilon_{5}^{(g)}\sqrt{n})\\\geq (1-e^{-\epsilon_{2}^{(m)} m})
P(\min_{\lambda_{\nu}\geq 0}\max_{\|\w\|_2=1}(\g^T\w-\lambda_{\nu}f(\w))\geq (1+\epsilon_{1}^{(m)})\sqrt{m}+\epsilon_{5}^{(g)}\sqrt{n}).\label{eq:probanal2}
\end{multline}

We now look at the left-hand side of the inequality in (\ref{eq:problemma}).
\begin{equation}
P(\min_{\nu\in R^n\setminus 0}\max_{\|\w\|_2=1}(-\nu^T A\w +\|\nu\|_2 g-\zeta_{\w,\nu})\geq 0)=
P(\min_{\nu\in R^n\setminus 0}\max_{\|\w\|_2=1}(-\nu^T A\w -f(\w)+\|\nu\|_2(g-\epsilon_{5}^{(g)}\sqrt{n}))\geq 0).\label{eq:leftprobanal0}
\end{equation}
Since $P(g\leq\epsilon_{5}^{(g)}\sqrt{n})< 1-e^{-\epsilon_{6}^{(g)} n}$ (where $\epsilon_{6}^{(g)}$ is, as all other $\epsilon$'s in this paper are, independent of $n$) from (\ref{eq:leftprobanal0}) we have
\begin{equation}
P(\min_{\nu\in R^n\setminus 0}\max_{\|\w\|_2=1}(-\nu^T A\w +\|\nu\|_2 g-\zeta_{\w,\nu})\geq 0)\leq
(1-e^{-\epsilon_{6}^{(g)} n})P(\min_{\nu\in R^n\setminus 0}\max_{\|\w\|_2=1}(-\nu^T A\w -f(\w))\geq 0)+e^{-\epsilon_{6}^{(g)} n}.\label{eq:leftprobanal1}
\end{equation}
Connecting (\ref{eq:problemma}), (\ref{eq:probanal0}), (\ref{eq:probanal1}), (\ref{eq:probanal2}), (\ref{eq:leftprobanal0}), and (\ref{eq:leftprobanal1}) we obtain
\begin{multline}
P(\min_{\nu\in R^n\setminus 0}\max_{\|\w\|_2=1}(-\nu^T A\w -f(\w))> 0)\geq\\
\frac{(1-e^{-\epsilon_{2}^{(m)} m})}{(1-e^{-\epsilon_{6}^{(g)} n})}
P(\min_{\lambda_{\nu}\geq 0}\max_{\|\w\|_2=1}(\g^T\w-\lambda_{\nu}f(\w))\geq (1+\epsilon_{1}^{(m)})\sqrt{m}+\epsilon_{5}^{(g)}\sqrt{n})
+\frac{e^{-\epsilon_{6}^{(g)} n}}{(1-e^{-\epsilon_{6}^{(g)} n})}.\label{eq:leftprobanal2}
\end{multline}
Let
\begin{equation}
\xi(f,\g)=\min_{\lambda_{\nu}\geq 0}\max_{\|\w\|_2=1}(\g^T\w-\lambda_{\nu}f(\w))\label{eq:xiDdef}
\end{equation}
and
\begin{equation}
\xi_D(f)=E\xi(f,\g).\label{eq:xidef}
\end{equation}
Using (\ref{eq:taudef1det4}) and (\ref{eq:xiDdef}), (\ref{eq:leftprobanal2}) becomes
\begin{equation}
P(-\tau(f,A)> 0)\geq
\frac{(1-e^{-\epsilon_{2}^{(m)} m})}{(1-e^{-\epsilon_{6}^{(g)} n})}
P(\xi(f,\g)\geq (1+\epsilon_{1}^{(m)})\sqrt{m}+\epsilon_{5}^{(g)}\sqrt{n})
+\frac{e^{-\epsilon_{6}^{(g)} n}}{(1-e^{-\epsilon_{6}^{(g)} n})}.\label{eq:leftprobanal3}
\end{equation}
Now we will make assumption that $\xi(f,\g)$ concentrates around $\xi_D(f)$ and that $\xi_D(f) \sim \sqrt{n}$, i.e.
\begin{equation}
P(|\xi(f,\g)-\xi_D(f)|\geq \epsilon_1^{(\xi)}\xi_D(f))\leq e^{-\epsilon_2^{(\xi)}\xi_D(f)}.\label{eq:xiconc}
\end{equation}
(This assumption can be avoided; however in the interest of maintaining as simple a presentation as possible we will state it).
Moreover, let
\begin{equation}
(1-\epsilon_1^{(\xi)})\xi_D(f)\geq (1+\epsilon_{1}^{(m)})\sqrt{m}+\epsilon_{5}^{(g)}\sqrt{n}.\label{eq:xiconc1}
\end{equation}
Then from (\ref{eq:leftprobanal3}) we have
\begin{equation}
P(-\tau(f,A)> 0)\geq
\frac{(1-e^{-\epsilon_{2}^{(m)} m})}{(1-e^{-\epsilon_{6}^{(g)} n})}(1-e^{-\epsilon_2^{(\xi)}\xi_D(f)})
+\frac{e^{-\epsilon_{6}^{(g)} n}}{(1-e^{-\epsilon_{6}^{(g)} n})}.\label{eq:leftprobanal4}
\end{equation}
Finally, if all assumptions we made indeed hold then
\begin{equation}
\lim_{n\rightarrow\infty}P(\tau(f,A)< 0) \geq \lim_{n\rightarrow\infty} \frac{(1-e^{-\epsilon_{2}^{(m)} m})}{(1-e^{-\epsilon_{6}^{(g)} n})}(1-e^{-\epsilon_2^{(\xi)}\xi_D(f)})
+\frac{e^{-\epsilon_{6}^{(g)} n}}{(1-e^{-\epsilon_{6}^{(g)} n})}=1.\label{eq:finaltau}
\end{equation}
In other words, if (\ref{eq:taudef1det3}) (or (\ref{eq:taudef1det31})), (\ref{eq:xiconc}), and (\ref{eq:xiconc1}) hold then for large $n$ one has with overwhelming probability that the random subspace of $\w$'s, $A\w$, will intersect set $S$ on the unit sphere.

We are now in position to state the following theorem which in a way complements Theorem \ref{thm:Gordonmesh}.

\begin{theorem}(Trapped in a mesh)
\label{thm:reversemesh} Let $m$ and $n$ be large and $m<n$ but proportional to $n$. Let $S$ be a subset of the unit Euclidean
sphere $S^{n-1}$ in $R^{n}$. Moreover, let $S$ be such that it can be characterized through a function $f(\w)$ in the following way
\begin{equation}
S=\{\w| ||\w||_2=1, f(\w)\leq 0\}.\label{eq:setSdefthm}
\end{equation}
Let $Y$ be a random
$(n-m)$-dimensional subspace of $R^{n}$, distributed uniformly in
the Grassmanian with respect to the Haar measure. For example, let
\begin{equation}
Y=\{\w| A\w=0\},\label{eq:setYdefthm}
\end{equation}
where $A$ is an $m\times n$ matrix of i.i.d. standard normals. Let $\g$ be an $m \times 1$ vector of i.i.d standard normals.
Further let
\begin{equation}
\xi_D(f)=E\min_{\lambda_{\nu}\geq 0}\max_{\|\w\|_2=1}(\g^T\w-\lambda_{\nu}f(\w)).\label{eq:xiDdefthm}
\end{equation}
Assume that $f(\w)$ is such that (\ref{eq:taudef1det3}) (or (\ref{eq:taudef1det31})) and (\ref{eq:xiconc}) hold.

1) Let $\epsilon_1$ and $\epsilon_2$ be arbitrarily small constants and let $m$ be such that
\begin{equation}
\xi_D(f)\geq (1+\epsilon_1)\sqrt{m}+\epsilon_2\sqrt{n}.\label{eq:xiconc1thm}
\end{equation}
Then
\begin{equation}
\lim_{n\rightarrow \infty}P(Y\cap S\neq 0)=1.
\label{eq:reversethmesh}
\end{equation}

2) On the other hand, let $m$ be such that
\begin{equation}
\xi_D(f) < \sqrt{m}-\frac{1}{4\sqrt{m}}.\label{eq:xiconc2thm}
\end{equation}
Then
\begin{equation}
\lim_{n\rightarrow \infty}P(Y\cap S = 0)=1.
\label{eq:directthmesh}
\end{equation}
\end{theorem}
\begin{proof}
The first part follows from the discussion presented above. For the second part we first observe from (\ref{eq:widthdef1det5})
and (\ref{eq:xiDdef})
\begin{equation}
w(f,\g)\leq\min_{\lambda_{\nu}\geq 0}\max_{\|\w\|_2=1}(\g^T\w-\lambda_{\nu}f(\w))=\xi(f,\g).\label{eq:wxiproofthmrev}
\end{equation}
Then we have
\begin{equation}
w_D(S)=w_D(f)=Ew(f,\g)\leq E\min_{\lambda_{\nu}\geq 0}\max_{\|\w\|_2=1}(\g^T\w-\lambda_{\nu}f(\w))=E\xi(f,\g)=\xi_D(f).\label{eq:wxiproofthmrev1}
\end{equation}
A combination of (\ref{eq:wxiproofthmrev1}) and the condition given in (\ref{eq:xiconc2thm}) gives
\begin{equation}
w_D(S)=w_D(f)\leq\xi_D(f)< \sqrt{m}-\frac{1}{4\sqrt{m}}.\label{eq:wxiproofthmrev2}
\end{equation}
which is then enough to apply Theorem \ref{thm:Gordonmesh} and obtain (\ref{eq:directthmesh}).
\end{proof}

In essence the above theorem provides a characterization of sets $S$ for which one can determine in a sense an optimal maximal/minimal dimension of the missing/intersecting subspace $A\w$. Of course the result of the previous theorem will be useful as long as one is able to handle (compute) $\xi_D$. Also, one should note that there are numerous other ways that can be used to present the main results obtained above. We chose the way given in the above theorem in order to be as close as possible to the original formulation given in Theorem \ref{thm:Gordonmesh} and at the same time to maintain a presentation that would in a way hint what the main ideas behind the entire mechanism are. For example, among many alternative formulations, the following two are probably even more natural than the version presented in the above theorem. First, instead of trying to formulate results along the lines of Theorem \ref{thm:Gordonmesh} one can formulate probabilistic results based on (\ref{eq:leftprobanal4}) and the corresponding ones that can be obtained in analogous way for $w(f,\g)$. We skip this exercise but do mention that in the absence of Theorem \ref{thm:Gordonmesh} such a presentation would be our preferable one. Second, instead of relying on quantity $\xi_D$ one can rely on the original $w_D(S)$. Since this modification is relatively simple we will provide a brief sketch of it below. We also do mention that this modification will in the end produce results that are visually more similar to the ones given in the original formulation in Theorem \ref{thm:Gordonmesh}. However, to achieve a mere similarity one is in a way forced to remodel formulations given in Theorem \ref{thm:reversemesh} which in our view contain a bit of a flavor as to how the entire mechanism works. That way one ultimately produces a visual analogue to Theorem \ref{thm:Gordonmesh} but at the expense of losing a bit of the hint as to what the core of the presented concept is. Still, we do believe that it is convenient to have such a formulation handy and we therefore present it below.

%%%%%%%%%%%%%%%%%%%%%%%%%%%%%%%%%%%%%%%%%%%%%%%%%%%%%%%%%%%
\subsection{An alternative formulation} \label{sec:alter}
%%%%%%%%%%%%%%%%%%%%%%%%%%%%%%%%%%%%%%%%%%%%%%%%%%%%%%%%%%%

The inequality given in (\ref{eq:probanal2}) can be further extended in the following way
\begin{multline}
P(\min_{\lambda_{\nu}\geq 0}\max_{\|\w\|_2=1}(\g^T\w-\lambda_{\nu}f(\w))\geq \|\h\|_2+\epsilon_{5}^{(g)}\sqrt{n})\\\geq (1-e^{-\epsilon_{2}^{(m)} m})
P(\min_{\lambda_{\nu}\geq 0}\max_{\|\w\|_2=1}(\g^T\w-\lambda_{\nu}f(\w))\geq (1+\epsilon_{1}^{(m)})\sqrt{m}+\epsilon_{5}^{(g)}\sqrt{n})\\ \geq
(1-e^{-\epsilon_{2}^{(m)} m})P(\max_{\|\w\|_2=1}\min_{\lambda_{\nu}\geq 0}(\g^T\w-\lambda_{\nu}f(\w))\geq (1+\epsilon_{1}^{(m)})\sqrt{m}+\epsilon_{5}^{(g)}\sqrt{n})\\
=(1-e^{-\epsilon_{2}^{(m)} m})P(-\min_{\|\w\|_2=1}\max_{\lambda_{\nu}\geq 0}(-\g^T\w+\lambda_{\nu}f(\w))\geq (1+\epsilon_{1}^{(m)})\sqrt{m}+\epsilon_{5}^{(g)}\sqrt{n}).\label{eq:probanal2alt}
\end{multline}
Using (\ref{eq:widthdef1det3}) one then has
\begin{multline}
P(\min_{\lambda_{\nu}\geq 0}\max_{\|\w\|_2=1}(\g^T\w-\lambda_{\nu}f(\w))\geq \|\h\|_2+\epsilon_{5}^{(g)}\sqrt{n})\\\geq (1-e^{-\epsilon_{2}^{(m)} m})P(w(f,\g)\geq (1+\epsilon_{1}^{(m)})\sqrt{m}+\epsilon_{5}^{(g)}\sqrt{n}).\label{eq:probanal3alt}
\end{multline}
Connecting (\ref{eq:problemma}), (\ref{eq:probanal0}), (\ref{eq:probanal1}), (\ref{eq:probanal2}), (\ref{eq:leftprobanal0}), (\ref{eq:leftprobanal1}), and (\ref{eq:probanal3alt}) we arrive at the following analogue to (\ref{eq:leftprobanal2})
\begin{multline}
P(\min_{\nu\in R^n\setminus 0}\max_{\|\w\|_2=1}(-\nu^T A\w -f(\w))> 0)\geq\\
\frac{(1-e^{-\epsilon_{2}^{(m)} m})}{(1-e^{-\epsilon_{6}^{(g)} n})}
P(w(f,\g)\geq (1+\epsilon_{1}^{(m)})\sqrt{m}+\epsilon_{5}^{(g)}\sqrt{n})
+\frac{e^{-\epsilon_{6}^{(g)} n}}{(1-e^{-\epsilon_{6}^{(g)} n})},\label{eq:leftprobanal4alt}
\end{multline}
which after using (\ref{eq:taudef1det4}) becomes the following analogue to (\ref{eq:leftprobanal3})
\begin{equation}
P(-\tau(f,A)> 0)\geq
\frac{(1-e^{-\epsilon_{2}^{(m)} m})}{(1-e^{-\epsilon_{6}^{(g)} n})}
P(w(f,\g)\geq (1+\epsilon_{1}^{(m)})\sqrt{m}+\epsilon_{5}^{(g)}\sqrt{n})
+\frac{e^{-\epsilon_{6}^{(g)} n}}{(1-e^{-\epsilon_{6}^{(g)} n})}.\label{eq:leftprobanal5alt}
\end{equation}
As in the previous subsection, we will make assumption that $w(f,\g)$ concentrates around $w_D(S)=w_D(f)=Ew(f,\g)$ (which is a bit easier to insure than the concentration of $\xi(f,\g)$; a way for doing so can be deduced from \cite{Gordon88}) and that $w_D(S)=w_D(f) \sim \sqrt{n}$, i.e.
\begin{equation}
P(|w(f,\g)-w_D(f)|\geq \epsilon_1^{(w)}w_D(f))\leq e^{-\epsilon_2^{(w)}w_D(f)}.\label{eq:wconcalt}
\end{equation}
(The assumption can also be avoided; as mentioned above, one way to do so even for a fairly general $f$ is to follow the presentation of \cite{Gordon88}; however as was the case in the previous subsection, in the interest of maintaining as simple a presentation as possible we will simply assume (\ref{eq:wconcalt})).
Moreover, let
\begin{equation}
(1-\epsilon_1^{(w)})w_D(f)\geq (1+\epsilon_{1}^{(m)})\sqrt{m}+\epsilon_{5}^{(g)}\sqrt{n}.\label{eq:wconc1alt}
\end{equation}
Then from (\ref{eq:leftprobanal5alt}) we have
\begin{equation}
P(-\tau(f,A)> 0)\geq
\frac{(1-e^{-\epsilon_{2}^{(m)} m})}{(1-e^{-\epsilon_{6}^{(g)} n})}(1-e^{-\epsilon_2^{(w)}w_D(f)})
+\frac{e^{-\epsilon_{6}^{(g)} n}}{(1-e^{-\epsilon_{6}^{(g)} n})}.\label{eq:leftprobanal6alt}
\end{equation}
Finally, if all assumptions we made indeed hold then
\begin{equation}
\lim_{n\rightarrow\infty}P(\tau(f,A)< 0) \geq \lim_{n\rightarrow\infty} \frac{(1-e^{-\epsilon_{2}^{(m)} m})}{(1-e^{-\epsilon_{6}^{(g)} n})}(1-e^{-\epsilon_2^{(w)}w_D(f))})
+\frac{e^{-\epsilon_{6}^{(g)} n}}{(1-e^{-\epsilon_{6}^{(g)} n})}=1.\label{eq:finaltau}
\end{equation}
In other words, if (\ref{eq:taudef1det3}) (or (\ref{eq:taudef1det31})), (\ref{eq:wconcalt}), and (\ref{eq:wconc1alt}) hold then for large $n$ one has with overwhelming probability that the random subspace of $\w$'s, $A\w$, will intersect set $S$ on the unit sphere.

We are now in position to state the following theorem which is an alternative formulation of Theorem \ref{thm:reversemesh} and as Theorem \ref{thm:reversemesh} in a way complements Theorem \ref{thm:Gordonmesh}.

\begin{theorem}(Trapped in a mesh -- alternative)
\label{thm:reversemeshalt} Let $m$ and $n$ be large and $m<n$ but proportional to $n$. Let $S$ be a subset of the unit Euclidean
sphere $S^{n-1}$ in $R^{n}$. Moreover, let $S$ be such that it can be characterized through a function $f(\w)$ in the following way
\begin{equation}
S=\{\w| ||\w||_2=1, f(\w)\leq 0\}.\label{eq:setSdefthmalt}
\end{equation}
Let $Y$ be a random
$(n-m)$-dimensional subspace of $R^{n}$, distributed uniformly in
the Grassmanian with respect to the Haar measure. For example, let
\begin{equation}
Y=\{\w| A\w=0\},\label{eq:setYdefthmalt}
\end{equation}
where $A$ is an $m\times n$ matrix of i.i.d. standard normals. Let $\g$ be an $m \times 1$ vector of i.i.d standard normals.
Further let
\begin{equation}
w_D(S)=w_D(f)=E\max_{\|\w\|_2=1,f(\w)\leq 0}\g^T\w=E\max_{\w\in S}\g^T\w .\label{eq:wDdefthmalt}
\end{equation}
Assume that $f(\w)$ is such that (\ref{eq:taudef1det3}) (or (\ref{eq:taudef1det31})) and (\ref{eq:wconcalt}) hold.

1) Let $\epsilon_1$ and $\epsilon_2$ be arbitrarily small constants and let $m$ be such that
\begin{equation}
w_D(S)=w_D(f)\geq (1+\epsilon_1)\sqrt{m}+\epsilon_2\sqrt{n}.\label{eq:wconc1thmalt}
\end{equation}
Then
\begin{equation}
\lim_{n\rightarrow \infty}P(Y\cap S\neq 0)=1.
\label{eq:reversethmeshalt}
\end{equation}

2) On the other hand, let $m$ be such that
\begin{equation}
w_D(S)=w_D(f) < \sqrt{m}-\frac{1}{4\sqrt{m}}.\label{eq:wconc2thmalt}
\end{equation}
Then
\begin{equation}
\lim_{n\rightarrow \infty}P(Y\cap S = 0)=1.
\label{eq:directthmeshalt}
\end{equation}
\end{theorem}
\begin{proof}
The first part follows from the discussion presented above. The second part follows from Theorem \ref{thm:Gordonmesh} and parts of its proof given in \cite{Gordon88}.
\end{proof}

Visually speaking, Theorem \ref{thm:reversemeshalt} may seem as a more natural complement to Theorem \ref{thm:Gordonmesh}. It is probably even a bit simpler than the formulation given in Theorem \ref{thm:reversemesh}. On the other hand, formulation in Theorem \ref{thm:reversemesh} is still our preferable one. In a way, it contains a bit of a description of what really is the key to success of the entire mechanism. If one is to give only the second portion of these theorems we do believe that then Theorem \ref{thm:reversemeshalt} is a more suitable choice (of course, by no surprise that is exactly what was done in \cite{Gordon88}).

%%%%%%%%%%%%%%%%%%%%%%%%%%%%%%%%%%%%%%%%%%%%%%%%%%%%%%%%%%%
\subsection{Comments} \label{sec:comments}
%%%%%%%%%%%%%%%%%%%%%%%%%%%%%%%%%%%%%%%%%%%%%%%%%%%%%%%%%%%

As far as understanding of the above theorems goes, there are several comments that we believe are in place. Below are some of them.

\begin{enumerate}
\item As one compares the statements of Theorems \ref{thm:reversemesh} and \ref{thm:reversemeshalt} on one side and the statement of Theorem \ref{thm:Gordonmesh} on the other it is clear that the concentration results are stated differently. In fact, not only are they stated differently they are also way inferior in Theorems \ref{thm:reversemesh} and \ref{thm:reversemeshalt}. We did mention right after Theorem \ref{thm:Gordonmesh} that determining concentrating constants is to the best of our knowledge an open problem even in the original formulation given in Theorem \ref{thm:Gordonmesh}. The same remains true for both of our theorems. The difference though is that while constants in Theorem \ref{thm:Gordonmesh} are most likely not the best possible ones, they are, when compared to generic $\epsilon$'s (given in our theorems), much better. We do mention that in this paper our major concern was a general type of result that relates to relation between $w_D(S)$ ($\xi_D$) and $m$ rather than a precise concentration analysis. Still, it would be of great importance if one could provide a way more precise analysis and determine ultimate optimality of concentrating constants as well. Our $\epsilon$'s can relatively easily be translated into concrete numbers. However, determining their optimal values is actually what requires a more careful approach. In fact, quite possibly, one may end up obtaining the optimal constants which are very large (simply, because one would have to encompass the entire family of sets $S$; such is the standard set by the generality of some of results presented in Theorems \ref{thm:Gordonmesh}, \ref{thm:reversemesh}, and \ref{thm:reversemeshalt}!). This is partially the reason why we haven't stated any specific constants but rather left such a problem to be solved on individual case basis.
\item Another important question that may arise based on our presentation is which of many alternative formulations would be the best possible. Answering such a question seems rather hard. Our experience is that when the mechanism works then typically everything (every quantity of interest) concentrates and if one is then fine with ignoring specifics of concentrations then essentially all formulations are fine.
\item The results presented above will not hold for all sets $S$. The question then remains can one determine the class of sets $S$ for which they will hold (such a subclass is determined by the two above theorems).
\item How hard is for a function to actually satisfy the assumptions that we have made? This is again a very generic question and it seems that it is better to form a class of functions for which they do hold, instead of trying to exclude those for which they do not.
\item How limiting/general are our descriptions of set $S$? In reality the description of set $S$ that we assumed is rather simple. We basically assumed that the entire set can be characterized through a functional inequality. However, our assumption was made mostly for the exposition purposes. The entire mechanism would go through as well even if set $S$ was characterized by an arbitrary number, say $L$, of functional inequalities, i.e., $f^{(l)}(\w)\leq 0,1,2\dots,L$.

\end{enumerate}

%%%%%%%%%%%%%%%%%%%%%%%%%%%%%%%%%%%%%%%%%%%%%%%%%%%%%%%%%%%
\section{Discussion -- how all of it actually works} \label{sec:discussion}
%%%%%%%%%%%%%%%%%%%%%%%%%%%%%%%%%%%%%%%%%%%%%%%%%%%%%%%%%%%

While the results presented in the previous section may seem a bit dry they are actually quite powerful. However, to really get a feeling how powerful they are one would have to convince himself/herself that there are scenarios when they can be used. While conceptually we discovered an array of sets $S$ for which subspace dimension results of Theorem \ref{thm:Gordonmesh} eventually through Theorems \ref{thm:reversemesh} and \ref{thm:reversemeshalt} become optimal we believe that it is easier to grasp the concept on small examples. Of course that is the reason why in the first part of the paper we briefly presented a problem that we were able to attack to full optimality using the mechanisms formulated in Theorems \ref{thm:reversemesh} and \ref{thm:reversemeshalt}. Below we will briefly sketch how the results presented in Section \ref{sec:back} actually fit into the context of the machinery presented in the previous section. Before doing so we just provide a small example that shows how the entire machinery can be modified a bit if function $f(\w)$ is of a special type.

%%%%%%%%%%%%%%%%%%%%%%%%%%%%%%%%%%%%%%%%%%%%%%%%%%%%%%%%%%%
\subsection{Homogeneous $f(\w)$} \label{sec:homf}
%%%%%%%%%%%%%%%%%%%%%%%%%%%%%%%%%%%%%%%%%%%%%%%%%%%%%%%%%%%

When function $f(\w)$ is homogeneous one can actually change a bit the presentation described above. In fact the presentation can be changed in many other scenarios as well; however we selected this one just to give a flavor as to what are possible options. Another reason is that sketching how the results given in Section \ref{sec:back} fit into what was presented above will be a bit easier. Now, let $f(\w)$ be a homogeneous function. Namely, let $f(\w)$ be such that
\begin{equation}
f(a\w)=a^df(\w),\label{eq:homf1}
\end{equation}
for any $a>0$ and a $d>0$. Then we say that function $f(\w)$ is positive homogeneous of degree $d$. Then for all practical purposes one can redefine $\tau(f,A)$ from (\ref{eq:taudef1det1}) in the following way
\begin{eqnarray}
\tau^{(h)}(f,A) =  \min & & f(\w) \nonumber \\
\mbox{subject to} & &  A\w=0 \nonumber \\
& & \|\w\|_2\leq 1.\label{eq:taudef1det1hom}
\end{eqnarray}
Proceeding then as in Section \ref{sec:detview} one can write
\begin{equation}
\tau^{(h)}(f,A) =  \min_{||\w||_2\leq 1}\max_{\nu} f(\w)+\nu^TA\w,\label{eq:taudef1det2hom}
\end{equation}
and assume that the structure of set $S$ is determined by a function $f(\w)$ for which it also holds
\begin{eqnarray}
\tau^{(h)}(f,A) & = & \min_{||\w||_2\leq 1}\max_{\nu} f(\w)+\nu^TA\w\nonumber \\
& = &  \max_{\nu} \min_{||\w||_2\leq 1}f(\w)+\nu^TA\w.\label{eq:taudef1det3hom}
\end{eqnarray}
If (as in Section \ref{sec:detview}) one instead focuses only on the sign of $\tau^{(h)}(f,A)$ one can relax a bit requirement (\ref{eq:taudef1det3hom}) to
\begin{eqnarray}
\mbox{sign}(\tau^{(h)}(f,A)) & = & \mbox{sign}(\min_{||\w||_2\leq 1}\max_{\nu} f(\w)+\nu^TA\w)\nonumber \\
& = &  \mbox{sign}(\max_{\nu} \min_{||\w||_2\leq 1}f(\w)+\nu^TA\w).\label{eq:taudef1det31hom}
\end{eqnarray}
After rearranging (\ref{eq:taudef1det3hom}) a bit we have
\begin{equation}
-\tau^{(h)}(f,A) = \min_{\nu} \max_{||\w||_2\leq 1} -f(\w)-\nu^TA\w,\label{eq:taudef1det4hom}
\end{equation}
and after rearranging (\ref{eq:taudef1det31hom}) a bit we have
\begin{equation}
-\mbox{sign}(\tau^{(h)}(f,A)) = \mbox{sign}(\min_{\nu} \max_{||\w||_2\leq 1} -f(\w)-\nu^TA\w).\label{eq:taudef1det41hom}
\end{equation}
Now one can repeat all the derivations from Section \ref{sec:probview} with $\tau^{(h)}(f,A)$ instead of $\tau(f,A)$. As a final result one would wind up with the theorems that are exactly the same as Theorems \ref{thm:reversemesh} and \ref{thm:reversemeshalt}. The only difference is that the assumptions on $f(\w)$ would be those from (\ref{eq:taudef1det3hom}) (or (\ref{eq:taudef1det31hom})) instead of those from (\ref{eq:taudef1det3}) (or (\ref{eq:taudef1det31})). This is a bit convenient since it essentially boils down to a duality over a convex set. Of course, everything we mentioned in this subsection remains true for any function for which the sign of $\tau(f,A)$ from (\ref{eq:taudef1det1}) does not change if one relaxes the sphere condition to the ball condition.

%%%%%%%%%%%%%%%%%%%%%%%%%%%%%%%%%%%%%%%%%%%%%%%%%%%%%%%%%%%
\subsection{An example of set $S$ where everything works} \label{sec:exmpl}
%%%%%%%%%%%%%%%%%%%%%%%%%%%%%%%%%%%%%%%%%%%%%%%%%%%%%%%%%%%

In this subsection we sketch how the results presented in Section \ref{sec:back} fit into the framework given in Section \ref{sec:probview}. We recall first that the problem that we were interested in in Section \ref{sec:back} is essentially the following: for a given $n$-dimensional $k$-sparse vector $\tilde{\x}$ (with say last $n-k$ components being zero) can one estimate the dimension of matrices $A$ in (\ref{eq:system}) such that the solution of (\ref{eq:l1}) is actually $k$-sparse. In fact let us be a bit more specific. Let us look at a $k$-sparse vector $\tilde{\x}$ (given the statistical structure  that will be later on assumed on $A$, one can without a loss of generality, set $\tilde{\x}_i=0, i=k+1,k+2,\dots,n$). Now, the question of interest is: given $A$ and $A\tilde{\x}$ (where $A$ is an $m\times n$ matrix and is typically called the measurement matrix) can one find $\x$ such that
\begin{equation}
A\x=A\tilde{\x}.\label{eq:systemdisc}
\end{equation}
To make sure that we maintain consistency we do emphasize that $A\tilde{\x}$ in (\ref{eq:systemdisc}) is what $\y$ in Section \ref{sec:back} is (in other words, although we did not state it anywhere in Section \ref{sec:back}, $y$ was essentially implied to be constructed as the product of matrix $A$ and a $k$-sparse vector $\x$). As we have mentioned in Section \ref{sec:back} a popular way to attack the above problem is to solve (\ref{eq:l1}), i.e. the following optimization problem
\begin{eqnarray}
\min_{\x} & & \|\x\|_1\nonumber \\
\mbox{subject to} & & A\x =  A\tilde{\x}.\label{eq:l1disc}
\end{eqnarray}
While the original problem (\ref{eq:systemdisc}) is NP-hard in the worst case, the optimization problem in (\ref{eq:l1disc}) is clearly solvable in polynomial time. Let $\hat{\x}$ be the solution of (\ref{eq:l1disc}). The question then is how often (if ever) $\hat{\x}=\tilde{\x}$. The line of thought first goes through the recognition that $\hat{\x}$ will be $k$-sparse $\tilde{\x}$ only if there is no $\w$ such that $\hat{\x}=\x+\w$, where $\w$ is in the null space of $A$ and satisfies (see, e.g. \cite{StojnicICASSP09,StojnicCSetam09,StojnicUpper10})
\begin{equation}
-\sum_{i=1}^{k}\w_i\geq \sum_{i=k+1}^{n}\|\w_i\|.
\end{equation}
If one then defines set $S$ on the unit sphere $S^{(n-1)}$ based on this parametrization of non-favorable $\w$'s one effectively obtains
\begin{equation}
S=\{\w| \sum_{i=1}^{k}\w_i+\sum_{i=k+1}^{n}\|\w_i\|\leq 0, \|\w\|_2=1\}.\label{eq:setSdisc}
\end{equation}
If one then defines $f(\w)$ as
\begin{equation}
f(\w)=\sum_{i=1}^{k}\w_i+\sum_{i=k+1}^{n}\|\w_i\|,\label{eq:fdisc}
\end{equation}
then clearly we have
\begin{equation}
S=\{\w| f(\w)\leq 0, \|\w\|_2=1\},\label{eq:setSdisc1}
\end{equation}
which fits into the description of $S$ given in (\ref{eq:setSdef}). Moreover, $S$ and ultimately $f$ will indeed satisfy all assumptions that we have made. Namely, $f(\w)$ from (\ref{eq:fdisc}) is positive homogeneous of degree $1$ and duality in (\ref{eq:taudef1det3hom}) and (\ref{eq:taudef1det31hom}) will easily hold. Also, let (as in Theorems \ref{thm:reversemesh} and \ref{thm:reversemeshalt})
\begin{equation}
Y=\{\w| \w\in R^n, A\w=0\}.\label{eq:setYdefdisc}
\end{equation}
Now, if one look at all $\w$'s from the null-space of $A$, i.e. at set $Y$, one can then connect the intersection of sets $Y$ and $S$ with $\tilde{\x}$ being equal or not to $\hat{\x}$. Namely, if $Y\cap S=0$ then $\tilde{\x}=\hat{\x}$ and if $Y\cap S\neq 0$ then there will be an $\hat{\x}$ such that $\hat{\x}_i=0,i=k+1,k+2,\dots,n$ and $\tilde{\x}\neq\hat{\x}$. Now, if one views the problem in a random context with matrix $A$ being an $m\times n$ matrix of i.i.d. standard normals, then one can for a given ratio $\frac{k}{n}$ determine the critical value of ratio $\frac{m}{n}$, $\frac{m_w}{n}=\frac{w_D(S)^2}{n}$, so that for $\frac{m}{n}>\frac{m_w}{n}=\frac{w_D(S)^2}{n}$ with overwhelming probability $\hat{\x}=\tilde{\x}$ for all $\tilde{\x}$ such that $\tilde{\x}_i=0,i=k+1,k+2,\dots,n$. On the other hand, for $\frac{m}{n}<\frac{m_w}{n}=\frac{w_D(S)^2}{n}$ with overwhelming probability there is an $\tilde{\x}$ such that $\tilde{\x}_i=0,i=k+1,k+2,\dots,n$ and $\hat{\x}\neq\tilde{\x}$ (in fact to be more in alignment with our theorems, instead of with overwhelming probability we should say with a probability that goes to one as $n\rightarrow\infty$).

Of course, what we presented above is just how critical $m_w$ can be connected to $w_D(S)$. In a way that solves only a half of the problem. The second half is to actually determine $w_D(S)$. That relates to question 1) that we mentioned in the short discussion after Theorem \ref{thm:Gordonmesh}. On the other hand, our main concern in here is question 2) from the very same discussion and along the same lines details related to handling $w_D(S)$ go beyond the scope of this paper. However, we do mention in passing that computing $w_D(S)$ was one of the problems of interest in \cite{StojnicCSetam09,StojnicUpper10} and the results obtained there are actually those presented in Theorem \ref{thm:thmweakthr}.

Also, what we presented in this section is a simple way how one can interpret the entire mechanism from previous sections when it comes to a particular set $S$. The interpretation given above is related to a rather simple set $S$. A more complicated version of $S$ where everything also works can be found in e.g. \cite{StojnicCSetamBlock09,StojnicUpperBlock10}.

%%%%%%%%%%%%%%%%%%%%%%%%%%%%%%%%%%%%%%%%%%%%%%%%%%%%%%%%%%%%%%%%%%%%%%%%%%%%%%%%
\section{Conclusion}
\label{sec:conc}
%%%%%%%%%%%%%%%%%%%%%%%%%%%%%%%%%%%%%%%%%%%%%%%%%%%%%%%%%%%%%%%%%%%%%%%%%%%%%%%%

In this paper we revisited a couple of classic probability results from \cite{Gordon88}. These results relate to the geometry of the intersection of random subspaces and subsets of the unit sphere in $R^n$ and properties of Gaussian processes. Namely, in \cite{Gordon88}, the likelihood of having random subspace of $R^n$ of dimension $n-m$ intersect a given set $S$ on the unit sphere was connected to a quantity describing set $S$ called the Gaussian width. Moreover, it was shown that $m$ can go (roughly speaking) as low as the squared gaussian width without having any significant likelihood of the random $n-m$-dimensional subspace intersecting set $S$. In this paper we provided a characterization of a class of sets $S$ for which if $m$ goes lower that the squared gaussian width of $S$ then it is highly likely that $n-m$-dimensional will intersect set $S$. In a way we provided a partial complement to the results of \cite{Gordon88}.

Also, to give a bit more flavor to a rather dry presentation of high dimensional geometry we gave a fairly detailed presentation of how the results that we created can in fact be utilized. We chose an example that deals with solving under-determined systems of linear equations with sparse solutions. It turns out that when the systems are random and gaussian the success of a technique called $\ell_1$-optimization when used to solve them can be connected to the problem of random subspaces intersecting given set $S$ on the unit sphere. We described how such a connection can be established and then provided a sketch as to how the main results of this paper actually work when such a connection is established.

While we presented only one specific example to give a flavor how everything practically works, the overall methodology is way more powerful. There are various other instances where we were able to successfully employ majority of the ideas presented here. Moreover, the mechanisms presented here are in fact a subcase of a much larger concept. In this paper though our focus were particular geometric results established in \cite{Gordon88} and how one can complement them. On the other hand, when viewed outside the scope of the results presented in \cite{Gordon88} our methodology admits consideration of substantially more general concepts. This goes way beyond the particular problems that we considered in this paper and we will present it elsewhere.

Finally, it is quite likely that Gordon's original results that we revisited here were only a tool towards much higher mathematical goals. Among them would immediately be a better version of the Dvoretzky theorem already established in Gordon's original work. Our results can then be used to complement all of such results where Gordon's estimates turned out to be of use. Of course, revisiting all of these takes a substantial effort that goes way beyond what we planned to present here. Here we only focused at the heart of the idea, which essentially boils down to simple reuse (with a little bit of our own recognition that duality theory can be quite powerful) of the Gordon's mechanism to prove its own optimality.

%\newpage1
%\setcounter{page}{1}
\begin{singlespace}
\bibliographystyle{plain}
\bibliography{GorEx}

\begin{thebibliography}{10}

\bibitem{ALPTJ09}
R.~Adamczak, A.~E. Litvak, A.~Pajor, and N.~Tomczak-Jaegermann.
\newblock Restricted isometry property of matrices with independent columns and
  neighborly polytopes by random sampling.
\newblock {\em Preprint}, 2009.
\newblock available at arXiv:0904.4723.

\bibitem{AS}
F.~Afentranger and R.~Schneider.
\newblock Random projections of regular simplices.
\newblock {\em Discrete Comput. Geom.}, 7(3):219--226, 1992.

\bibitem{Bar}
R.~Baraniuk, M.~Davenport, R.~DeVore, and M.~Wakin.
\newblock A simple proof of the restricted isometry property for random
  matrices.
\newblock {\em Constructive Approximation}, 28(3), 2008.

\bibitem{BorockyHenk}
K.~Borocky and M.~Henk.
\newblock Random projections of regular polytopes.
\newblock {\em Arch. Math. (Basel)}, 73(6):465--473, 1999.

\bibitem{Crip}
E.~Candes.
\newblock The restricted isometry property and its implications for compressed
  sensing.
\newblock {\em Compte Rendus de l'Academie des Sciences, Paris, Series I, 346},
  pages 589--59, 2008.

\bibitem{CRT}
E.~Candes, J.~Romberg, and T.~Tao.
\newblock Robust uncertainty principles: exact signal reconstruction from
  highly incomplete frequency information.
\newblock {\em IEEE Trans. on Information Theory}, 52:489--509, December 2006.

\bibitem{CWBreweighted}
E.~Candes, M.~Wakin, and S.~Boyd.
\newblock Enhancing sparsity by reweighted l1 minimization.
\newblock {\em J. Fourier Anal. Appl.}, 14:877--905, 2008.

\bibitem{SChretien08}
S.~Chretien.
\newblock An alternating ell-1 approach to the compressed sensing problem.
\newblock 2008.
\newblock available online at http://www.dsp.ece.rice.edu/cs/.

\bibitem{CoMu05}
G.~Cormode and S.~Muthukrishnan.
\newblock Combinatorial algorithms for compressed sensing.
\newblock {\em SIROCCO, 13th Colloquium on Structural Information and
  Communication Complexity}, pages 280--294, 2006.

\bibitem{DG08}
M.~E. Davies and R.~Gribonval.
\newblock Restricted isometry constants where ell-p sparse recovery can fail
  for $0 < p \leq 1$.
\newblock available online at http://www.dsp.ece.rice.edu/cs/.

\bibitem{DonohoUnsigned}
D.~Donoho.
\newblock Neighborly polytopes and sparse solutions of underdetermined linear
  equations.
\newblock 2004.
\newblock Technical report, Department of Statistics, Stanford University.

\bibitem{DonohoPol}
D.~Donoho.
\newblock High-dimensional centrally symmetric polytopes with neighborlines
  proportional to dimension.
\newblock {\em Disc. Comput. Geometry}, 35(4):617--652, 2006.

\bibitem{DT}
D.~Donoho and J.~Tanner.
\newblock Neighborliness of randomly-projected simplices in high dimensions.
\newblock {\em Proc. National Academy of Sciences}, 102(27):9452--9457, 2005.

\bibitem{DonohoSigned}
D.~Donoho and J.~Tanner.
\newblock Sparse nonnegative solutions of underdetermined linear equations by
  linear programming.
\newblock {\em Proc. National Academy of Sciences}, 102(27):9446--9451, 2005.

\bibitem{FL08}
S.~Foucart and M.~J. Lai.
\newblock Sparsest solutions of underdetermined linear systems via ell-q
  minimization for $0 < q \leq 1$.
\newblock available online at http://www.dsp.ece.rice.edu/cs/.

\bibitem{GiStTrVe06}
A.~Gilbert, M.~J. Strauss, J.~A. Tropp, and R.~Vershynin.
\newblock Algorithmic linear dimension reduction in the l1 norm for sparse
  vectors.
\newblock {\em 44th Annual Allerton Conference on Communication, Control, and
  Computing}, 2006.

\bibitem{GiStTrVe07}
A.~Gilbert, M.~J. Strauss, J.~A. Tropp, and R.~Vershynin.
\newblock One sketch for all: fast algorithms for compressed sensing.
\newblock {\em ACM STOC}, pages 237--246, 2007.

\bibitem{Gordon88}
Y.~Gordon.
\newblock On {M}ilman's inequality and random subspaces which escape through a
  mesh in ${R}^n$.
\newblock {\em Geometric Aspect of of functional analysis, Isr. Semin. 1986-87,
  Lect. Notes Math}, 1317, 1988.

\bibitem{GN03}
R.~Gribonval and M.~Nielsen.
\newblock Sparse representations in unions of bases.
\newblock {\em IEEE Trans. Inform. Theory}, 49(12):3320--3325, December 2003.

\bibitem{GN04}
R.~Gribonval and M.~Nielsen.
\newblock On the strong uniqueness of highly sparse expansions from redundant
  dictionaries.
\newblock In {\em Proc. Int Conf. Independent Component Analysis (ICA'04)},
  LNCS. Springer-Verlag, September 2004.

\bibitem{GN07}
R.~Gribonval and M.~Nielsen.
\newblock Highly sparse representations from dictionaries are unique and
  independent of the sparseness measure.
\newblock {\em {A}ppl. {C}omput. {H}arm. {A}nal.}, 22(3):335--355, May 2007.

\bibitem{PMM}
P.~McMullen.
\newblock Non-linear angle-sum relations for polyhedral cones and polytopes.
\newblock {\em Math. Proc. Cambridge Philos. Soc.}, 78(2):247--261, 1975.

\bibitem{Ruben}
H.~Ruben.
\newblock On the geometrical moments of skew regular simplices in
  hyperspherical space; with some applications in geometry and mathematical
  statistics.
\newblock {\em Acta. Math. (Uppsala)}, 103:1--23, 1960.

\bibitem{Ver}
M.~Rudelson and R.~Vershynin.
\newblock Geometric approach to error correcting codes and reconstruction of
  signals.
\newblock {\em International Mathematical Research Notices}, 64:4019 -- 4041,
  2005.

\bibitem{RVmesh}
M.~Rudelson and R.~Vershynin.
\newblock On sparse reconstruction from {F}ourier and {G}aussian measurements.
\newblock {\em Comm. on Pure and Applied Math.}, 61(8), 2007.

\bibitem{SCY08}
R.~Saab, R.~Chartrand, and O.~Yilmaz.
\newblock Stable sparse approximation via nonconvex optimization.
\newblock {\em ICASSP, IEEE Int. Conf. on Acoustics, Speech, and Signal
  Processing}, Apr. 2008.

\bibitem{SaZh08}
V.~Saligrama and M.~Zhao.
\newblock Thresholded basis pursuit: Quantizing linear programming solutions
  for optimal support recovery and approximation in compressed sensing.
\newblock 2008.
\newblock available on arxiv.

\bibitem{StojnicUpperBlock10}
M.~Stojnic.
\newblock Optimality of $\ell_2/\ell_1$-optimization block-length dependent
  thresholds.
\newblock available at arXiv.

\bibitem{StojnicEquiv10}
M.~Stojnic.
\newblock A rigorous geometry-probability equivalence in characterization of
  $\ell_1$-optimization.
\newblock available at arXiv.

\bibitem{StojnicUpper10}
M.~Stojnic.
\newblock Upper-bounding $\ell_1$-optimization weak thresholds.
\newblock available at arXiv.

\bibitem{StojnicCSetamBlock09}
M.~Stojnic.
\newblock Block-length dependent thresholds in block-sparse compressed sensing.
\newblock {\em submitted to IEEE Trans. on Information Theory}, 2009.
\newblock available at arXiv:0907.3679.

\bibitem{StojnicICASSP09}
M.~Stojnic.
\newblock A simple performance analysis of $\ell_1$-optimization in compressed
  sensing.
\newblock {\em ICASSP, International Conference on Acoustics, Signal and Speech
  Processing}, April 2009.

\bibitem{StojnicCSetam09}
M.~Stojnic.
\newblock Various thresholds for $\ell_1$-optimization in compressed sensing.
\newblock {\em submitted to IEEE Trans. on Information Theory}, 2009.
\newblock available at arXiv:0907.3666.

\bibitem{VS}
A.~M. Vershik and P.~V. Sporyshev.
\newblock Asymptotic behavior of the number of faces of random polyhedra and
  the neighborliness problem.
\newblock {\em Selecta Mathematica Sovietica}, 11(2), 1992.

\end{thebibliography}
\end{singlespace}

\end{document}